\documentclass{amsart}
\usepackage{amssymb}
\usepackage{amsmath}
\usepackage{amsthm}
\usepackage{mathrsfs}
\usepackage{colortbl}
\usepackage{graphicx}
\newtheorem{theorem}{Theorem}
\newtheorem{lemma}[theorem]{Lemma}

\newtheorem{remark}[theorem]{Remark}

\theoremstyle{definition}
\newtheorem{definition}[theorem]{Definition}

\makeatletter

\@addtoreset{theorem}{section}
\makeatother

\makeatletter

\@addtoreset{equation}{section}
\makeatother

\begin{document}                                                 
\title[Multiphase flow system with phase transition]{Energetic variational approaches for multiphase flow systems with phase transition}                                 
\author[Hajime Koba]{Hajime Koba}                                
\address{Graduate School of Engineering Science, Osaka University\\
1-3 Machikaneyamacho, Toyonaka, Osaka, 560-8531, Japan}                                  
%\curraddr{...}                                   
\email{iti@sigmath.es.osaka-u.ac.jp}

%\date{}                                      
\thanks{This work was partly supported by the Japan Society for the Promotion of Science (JSPS) KAKENHI Grant Number JP21K03326.}                                     
%\translator{...}
\keywords{Mathematical modeling, Multiphase flow, Phase transition, Surface tension, Interface}                            
\subjclass[]{70-10,76-10,76T10,35A15}

\begin{abstract}
We study the governing equations for the motion of the fluid particles near air-water interface from an energetic point of view. Since evaporation and condensation phenomena occur at the interface, we have to consider phase transition. This paper applies an energetic variational approach to derive multiphase flow systems with phase transition, where a multiphase flow means compressible and incompressible two-phase flow. We also research the conservation and energy laws of our system. The key ideas of deriving our systems are to acknowledge the existence of the interface and to apply an energetic variational approach. More precisely, we assume that both the coefficient of surface tension and the density of the interface are constants, and we apply an energetic variational approach to look for the dominant equations for the densities of our multiphase flow systems with phase transition. As applications, we can derive the usual Euler and Navier-Stokes systems, or a two-phase flow system with surface tension by our methods.
\end{abstract}
\maketitle

\section{Introduction}\label{sect1}

\begin{figure}[htbp]
\input{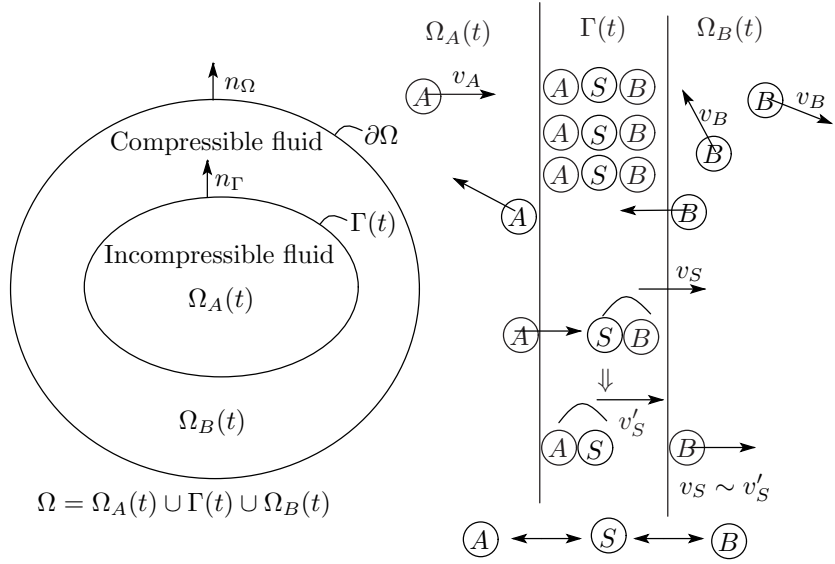}
\caption{Multiphase Flow and Phase Transition}
\label{Fig1}
\end{figure}

We are interested in the motion of the fluid particles near the boundary between the atmosphere and the ocean. We call the boundary the \emph{air-water interface}. Since evaporation and condensation phenomena occur at the interface, we have to study multiphase flow with phase transition in order to understand air-sea interaction. This paper considers the governing equations for the motion of the fluid particles in two moving domains and the interface from an energetic point of view. We employ an energetic variational approach to derive our multiphase flow systems with phase transition. Of course, this paper proposes our system as one of the models for phase transition.

Let us first introduce fundamental notations. Let $t \geq 0$ be the time variable, and $x( = { }^t (x_1 , x_2, x_3 ) ) \in \mathbb{R}^3$ the spatial variable. Fix $T >0$. Let $\Omega \subset \mathbb{R}^3$ be a bounded domain with a smooth boundary $\partial \Omega$. The symbol $n_\Omega = n_\Omega (x) = { }^t (n^\Omega_1 , n^\Omega_2 , n^\Omega_3 )$ denotes the unit outer normal vector at $x \in \partial \Omega$. Let $\Omega_A (t) (= \{ \Omega_A (t) \}_{0 \leq t < T} )$ be a bounded domain in $\mathbb{R}^3$ with a moving boundary $\Gamma (t)$. Assume that $\Gamma (t) (= \{ \Gamma (t) \}_{0 \leq t < T})$ is a smoothly evolving surface and is a closed Riemannian 2-dimensional manifold. The symbol $n_\Gamma = n_\Gamma ( x , t ) = { }^t (n^\Gamma_1 , n^\Gamma_2 , n^\Gamma_3)$ denotes the unit outer normal vector at $x \in \Gamma (t)$. For each $t \in [0,T)$, assume that $\Omega_A (t) \Subset \Omega$. Set $\Omega_B (t) = \Omega \setminus \overline{\Omega_A (t)}$. It is clear that $\Omega = \Omega_A (t) \cup \Gamma (t) \cup \Omega_B (t)$ (see Figure \ref{Fig1}). Set
\begin{multline}\label{eq11}
\Omega_{A,T} = \bigcup_{0< t < T} \{ \Omega_A (t) \times \{ t \} \},{ \ } \Omega_{B,T} = \bigcup_{0< t < T} \{ \Omega_B (t) \times \{ t \} \},\\
\Gamma_T = \bigcup_{0< t < T} \{ \Gamma (t) \times \{ t \} \},{ \ }\Omega_T = \Omega \times (0,T),{ \ }\partial \Omega_T = \partial \Omega \times (0,T).
\end{multline}

In this paper we assume that the fluid in $\Omega_{A,T}$ is an incompressible one, and that the fluid in $\Omega_{B,T}$ is a compressible one. Let us state physical notations. Let $\rho_A = \rho_A ( x , t)$, $v_A = v_A ( x , t) = { }^t (v^A_1 , v^A_2 , v^A_3 )$, $\pi_A = \pi_A (x,t)$, and $\mu_A = \mu_A (x,t)$ be the density, the velocity, the pressure, and the viscosity of the fluid in $\Omega_A (t)$, respectively. Let $\rho_B = \rho_B ( x , t)$, $v_B = v_B ( x , t) = { }^t (v^B_1 , v^B_2 , v^B_3 )$, $\pi_B = \pi_B (x,t)$, and $\mu_B = \mu_B (x,t)$, $\lambda_B = \lambda_B (x,t)$ be the density, the velocity, the pressure, and two viscosities of the fluid in $\Omega_B (t)$, respectively. Let $v_S = v_S (x , t) = { }^t (v^S_1 , v^S_2 , v^S_3 )$ be the motion velocity of the evolving surface $\Gamma (t)$. The symbol $\rho_0 >0$ denotes the density of the interface $\Gamma (t)$, and $\pi_0 \in \mathbb{R} \setminus \{ 0 \}$ denotes the surface tension (coefficient) at $x \in \Gamma (t)$. We assume that $\rho_A$, $v_A$, $\pi_A$, $\mu_A$, $\rho_B$, $v_B$, $\pi_B$, $\mu_B$, $\lambda_B$, $v_S$ are smooth functions in $\mathbb{R}^4$, and that $\rho_0$, $\pi_0$ are constants.

Let us explain the background and the ansatz of this study. Let us now consider the phase transition phenomenon on ice melting. There exists a layer between the ice and the air. The layer is called a \emph{quasi-liquid layer} (see Kuroda-Lacmann \cite{KL82}, Furukawa-Yamamoto-Kuroda \cite{FYK87}). It is well-known that a quasi-liquid layer has both liquid and solid properties. Experiments in Sazaki-Zepeda-Nakatsubo-Yokoyama-Furukawa \cite{SZNYF10} showed that ice particles change into particles in the quasi-liquid layer and then the particles in the layer change into water vapor. A similar process occurs when ice forms. Therefore, we can consider a two-phase problem with a phase transition as a three-phase problem. In this paper, we admit the existence of a surface mass at the interface $\Gamma(t)$, and assume that the particles at the interface can change into both particles in $\Omega_A (t)$ and particles in $\Omega_B (t)$ (see Figure \ref{Fig1}).

Let us explain the key restriction of mathematical modeling of multiphase flow systems with phase transition. We assume that
\begin{equation}\label{eq12}
\begin{cases}
{\rm{div}} v_A =0 & \text{ in }\Omega_{A,T},\\
v_B \cdot n_\Omega =0 & \text{ on }\partial \Omega_T,
\end{cases}{ \ }
\begin{cases}
v_A \cdot n_\Gamma = v_S \cdot n_\Gamma & \text{ on } \Gamma_T,\\
v_B \cdot n_\Gamma = v_S \cdot n_\Gamma  & \text{ on } \Gamma_T.
\end{cases}
\end{equation}
The condition ${\rm{div}} v_A =0$ means the incompressibility condition of the fluid in $\Omega_{A,T}$, and $v_B \cdot n_\Omega =0$ means that fluid particles do not go out of the domain $\Omega$. In general, we assume several jump conditions when we make models for multiphase flow with phase transition. This paper does not use assumptions on jump conditions to derive our systems. See Slattery-Sagis-Oh \cite{SSO07} for several jump conditions on interfacial phenomena.

This paper has three purposes. The first one is to make the following inviscid model for multiphase flow with phase transition:
\begin{equation}\label{eq13}
\begin{cases}
{\displaystyle{D_t^A \rho_A = {\rm{div}}\left\{ \frac{\rho_0}{\pi_0}\pi_A v_A \right\} }}& \text{ in } \Omega_{A,T},\\
{\displaystyle{D_t^B \rho_B + ({\rm{div}} v_B) \rho_B = {\rm{div}}\left\{ \frac{\rho_0}{\pi_0}\pi_B v_B \right\} }}& \text{ in } \Omega_{B,T},\\
{\displaystyle{\rho_A D_t^A v_A + {\rm{grad}}\pi_A  = - {\rm{div}}\left\{ \frac{\rho_0}{\pi_0}\pi_A v_A \right\} v_A }}& \text{ in } \Omega_{A,T},\\
{\displaystyle{\rho_B D_t^B v_B + {\rm{grad}}\pi_B  = - {\rm{div}}\left\{ \frac{\rho_0}{\pi_0}\pi_B v_B \right\} v_B }}& \text{ in } \Omega_{B,T},\\
\pi_0 H_\Gamma n_\Gamma - \pi_A n_\Gamma +  \pi_B n_\Gamma ={ }^t(0,0,0) & \text{ on } \Gamma_T
\end{cases}
\end{equation}
with \eqref{eq12}. Here $D_t^A f = \partial_t f + (v_A \cdot \nabla ) f $, $D_t^B f = \partial_t f + (v_B \cdot \nabla ) f $, $(v_A \cdot \nabla ) f = v^A_1 \partial_1 f + v^A_2 \partial_2 f + v^A_3 \partial_3 f$, $(v_B \cdot \nabla ) f = v^B_1 \partial_1 f + v^B_2 \partial_2 f + v^B_3 \partial_3 f$, ${\rm{div}}v_A = \nabla \cdot v_A$, ${\rm{div}}v_B = \nabla \cdot v_B$, ${\rm{grad}} f = \nabla f$, $\nabla = { }^t (\partial_1 , \partial_2 , \partial_3)$, $\partial_i = \partial/{\partial x_i}$, and $\partial_t = \partial/{\partial t}$. The symbol $H_\Gamma = H_\Gamma (x , t )$ denotes the \emph{mean curvature in the direction} $n_\Gamma$ defined by $H_\Gamma  = - {\rm{div}}_\Gamma n_\Gamma = - (\partial_1^\Gamma n_1^\Gamma + \partial_2^\Gamma n^\Gamma_2 + \partial_3^\Gamma n^\Gamma_3)$, where $\partial^\Gamma_i f := \sum_{j=1}^3(\delta_{ij} - n^\Gamma_i n^\Gamma_j ) \partial_j f = \partial_j f - n_j^\Gamma (n_\Gamma \cdot \nabla )f$. More precisely, under the restriction \eqref{eq12} we apply an energetic variational approach to derive system \eqref{eq13}. See subsection \ref{subsec42} for details. Remark that the motion velocity $v_S$ is given by
\begin{equation*}
v_S = \frac{1}{\pi_0 H_\Gamma}\{ \pi_A (v_A \cdot n_\Gamma) - \pi_B (v_B \cdot n_\Gamma) \} n_\Gamma \text{ on } \Gamma_T
\end{equation*}
if $H_\Gamma \neq 0$ and $P_\Gamma v_S = { }^t(0,0,0)$. See the proof of Theorem \ref{thm22} in Section \ref{sect3} for details.

The second one is to make the following viscous model for multiphase flow with phase transition:
\begin{equation}\label{eq14}
\begin{cases}
{\displaystyle{D_t^A \rho_A= - \frac{\rho_0}{\pi_0} {\rm{div}}( \mathcal{T}_A v_A) }}& \text{ in } \Omega_{A,T},\\
{\displaystyle{D_t^B \rho_B + ({\rm{div}} v_B) \rho_B  = - \frac{\rho_0}{\pi_0} {\rm{div}} ( \mathcal{T}_B v_B )}} & \text{ in } \Omega_{B,T},\\
{\displaystyle{\rho_A D_t^A v_A  =  {\rm{div}} \mathcal{T}_A + \frac{\rho_0}{\pi_0} {\rm{div}}( \mathcal{T}_A v_A) v_A }}& \text{ in } \Omega_{A,T},\\
{\displaystyle{\rho_B D_t^B v_B  =  {\rm{div}} \mathcal{T}_B + \frac{\rho_0}{\pi_0} {\rm{div}}( \mathcal{T}_B v_B)v_B  }}& \text{ in } \Omega_{B,T},\\
\pi_0 H_\Gamma n_\Gamma + \widetilde{\mathcal{T}}_A n_\Gamma - \widetilde{\mathcal{T}}_B n_\Gamma= { }^t( 0,0,0) & \text{ on } \Gamma_T
\end{cases}
\end{equation}
with
\begin{equation}\label{eq15}
\begin{cases}
{\rm{div}} v_A =0 & \text{ in }\Omega_{A,T},\\
v_A \cdot n_\Gamma = v_S \cdot n_\Gamma & \text{ on } \Gamma_T,\\
v_B \cdot n_\Gamma = v_S \cdot n_\Gamma  & \text{ on } \Gamma_T,
\end{cases}{ \ }\begin{cases}
v_B = { }^t (0,0,0) & \text{ in }\partial \Omega_T,\\
P_\Gamma v_A = { }^t (0,0,0) & \text{ on } \Gamma_T,\\
P_\Gamma v_B = { }^t (0,0,0) & \text{ on } \Gamma_T,
\end{cases}
\end{equation}
where
\begin{equation}\label{eq16}
\begin{cases}
\mathcal{T}_A  =\mathcal{T}_A (\pi_A , v_A) :=\mu_A D (v_A) - \pi_A I_{3 \times 3} ,\\
\mathcal{T}_B  = \mathcal{T}_B (\pi_B, v_B ) :=\mu_B D (v_B)  + \lambda_B ({\rm{div}} v_B)I_{3\times 3} - \pi_B I_{3 \times 3},\\ 
\widetilde{\mathcal{T}}_A =\widetilde{\mathcal{T}}_A (\pi_A , v_A) :=\mu_A (n_\Gamma \cdot (n_\Gamma \cdot \nabla ) v_A) - \pi_A ,\\
\widetilde{\mathcal{T}}_B  = \widetilde{\mathcal{T}}_B (\pi_B, v_B ) :=\mu_B (n_\Gamma \cdot (n_\Gamma \cdot \nabla ) v_B) + \lambda_B ({\rm{div}} v_B) - \pi_B . 
\end{cases}
\end{equation}
Here $D (v_A) =\{ { }^t(\nabla v_A) +(\nabla v_A) \}/2$ and $D (v_B) =\{ { }^t(\nabla v_B) +(\nabla v_B) \}/2$. The symbol $I_{3\times 3}$ denotes the $3 \times 3$ identity matrix, and $P_\Gamma = P_\Gamma (x ,t)$ the \emph{orthogonal projection to a tangent space} defined by $P_\Gamma = I_{3 \times 3} - n_\Gamma \otimes n_\Gamma$, where  $\otimes$ is the tensor product. More precisely, under the restriction \eqref{eq15} we apply an energetic variational approach to derive system \eqref{eq14}. See subsection \ref{subsec41} for details. Remark that the motion velocity $v_S$ is given by
\begin{equation*}
v_S = \frac{1}{\pi_0 H_\Gamma}\{ - \widetilde{\mathcal{T}}_A (v_A \cdot n_\Gamma) +  \widetilde{\mathcal{T}}_B (v_B \cdot n_\Gamma) \} n_\Gamma \text{ on } \Gamma_T
\end{equation*}
if $H_\Gamma \neq 0$ and $P_\Gamma v_S = { }^t(0,0,0)$. See the proof of Theorem \ref{thm22} in Section \ref{sect3} for details.
See the proof of Theorem \ref{thm22} in Section \ref{sect3} for details.

The third one is to investigate the conservation and energy laws, and the conservative form of our systems \eqref{eq13} and \eqref{eq14}. In fact, any solution to system \eqref{eq14} with \eqref{eq15} satisfies that for $t_1 < t_2$,
\begin{multline}\label{eq17}
\int_{\Omega_A (t_2)} \rho_A (x,t_2) { \ } d x + \int_{\Omega_B (t_2)} \rho_B (x,t_2) { \ } d x + \int_{\Gamma (t_2)} \rho_0 { \ }d \mathcal{H}_x^2 \\
= \int_{\Omega_A (t_1)} \rho_A (x,t_1) { \ } d x + \int_{\Omega_B (t_1)} \rho_B (x,t_1) { \ } d x + \int_{\Gamma (t_1)} \rho_0 { \ }d \mathcal{H}_x^2,
\end{multline}
and
\begin{multline}\label{eq18}
\int_{\Omega_A (t_2)} \frac{1}{2} \rho_A  \vert v_A  \vert^2{ \ }d x + \int_{\Omega_B (t_2)} \frac{1}{2} \rho_B  \vert v_B  \vert^2{ \ }d x + \int_{t_1}^{t_2 } \int_{\Omega_A (t)} \mu_A  \vert D (v_A) \vert^2 { \ }d x d t\\ + \int_{t_1}^{t_2} \int_{\Omega_B (t)} (\mu_B  \vert D (v_B) \vert^2 + \lambda_B  \vert { \rm{div} }v_B  \vert ^2){ \ }d x dt\\
= \int_{\Omega_A (t_1)} \frac{1}{2} \rho_A  \vert v_A  \vert^2{ \ }d x + \int_{\Omega_B (t_1)} \frac{1}{2} \rho_B  \vert v_B  \vert^2{ \ }d x\\
 + \int_{t_1}^{t_2} \int_{\Gamma (t)} ({\rm{div}}_\Gamma v_S )\pi_0 { \ }d \mathcal{H}^2_x d t + \int_{t_1}^{t_2 } \int_{\Omega_A (t)}\left( \frac{\rho_0}{2 \pi_0} {\rm{div}} ( \mathcal{T}_A v_A)  \vert v_A \vert^2\right) { \ }d x d t\\ + \int_{t_1}^{t_2} \int_{\Omega_B (t)} \left( ({\rm{div}} v_B) \pi_B + \frac{\rho_0}{2 \pi_0}  {\rm{div}}(\mathcal{T}_B v_B)  \vert v_B \vert^2\right) { \ }d x dt. 
\end{multline}
Moreover, any solution to system \eqref{eq13} with \eqref{eq12} satisfies \eqref{eq17} and \eqref{eq18} with $\mu_A \equiv \mu_B \equiv \lambda_B \equiv 0$. Here $ \vert D (v_A) \vert^2= D(v_A): D (v_A)$, $ \vert D (v_B) \vert^2= D(v_B): D (v_B)$, and $d \mathcal{H}^2_x$ denotes the \emph{2-dimensional Hausdorff measure}. We often call \eqref{eq17} and \eqref{eq18}, the \emph{law of conservation of mass} and the \emph{energy equality}, respectively. We easily check that system \eqref{eq14} with ${\rm{div}} v_A =0$ satisfies the following conservative form:
\begin{equation}\label{eq19}
\begin{cases}
{\displaystyle{\partial_t \rho_A + {\rm{div}}\left( \rho_A v_A + \frac{\rho_0}{\pi_0} \mathcal{T}_A v_A \right) = 0}}& \text{ in } \Omega_{A,T},\\
{\displaystyle{\partial_t \rho_B + {\rm{div}}\left( \rho_B v_B + \frac{\rho_0}{\pi_0} \mathcal{T}_B v_B \right) = 0}}& \text{ in } \Omega_{B,T},\\
{\displaystyle{\partial_t (\rho_A v_A) + {\rm{div}} \left( \rho_A v_A \otimes v_A - \mathcal{T}_A \right) = { }^t(0,0,0)}}& \text{ in } \Omega_{A,T},\\
{\displaystyle{\partial_t (\rho_B v_B) + {\rm{div}}\left( \rho_B v_B \otimes v_B - \mathcal{T}_B \right) = { }^t (0,0,0)}}& \text{ in } \Omega_{B,T}.
\end{cases}
\end{equation}

\begin{remark}\label{rem11}$(\rm{i})$ If we change \eqref{eq12} or \eqref{eq15} to another restriction, then we can derive another system by applying our approaches.\\
$(\rm{ii})$ If we choose $\rho_0 =0$, then we derive the usual Euler and Navier-Stokes systems, or a two-phase flow system with surface tension by our approaches.
\end{remark}

Let us explain three key ideas of deriving our multiphase flow systems with phase transition. The first point is to acknowledge the existence of the interface, that is, we assume that the density of the interface is a positive constant $\rho_0$. The second point is to divide the condition $(v_A - v_B) \cdot n_\Gamma = 0$ into $v_A \cdot n_\Gamma = v_S \cdot n_\Gamma$ and $v_B \cdot n_\Gamma = v_S \cdot n_\Gamma$. The third point is to make use of an energetic variational approach. More precisely, we apply an energetic variational approach in order to look for functions $\Phi_A$ and $\Phi_B$ satisfying
\begin{equation*}
\begin{cases}
D_t^A \rho_A  = \Phi_A & \text{ in } \Omega_{A,T},\\
D_t^B \rho_B + ({\rm{div}} v_B) \rho_B = \Phi_B & \text{ in } \Omega_{B,T},
\end{cases}
\end{equation*}
and
\begin{equation*}
\frac{d }{d t} \left( \int_{\Omega_A (t)} \rho_A (x,t) { \ }d x + \int_{\Omega_B(t)} \rho_B (x,t) { \ }d x + \int_{\Gamma (t)} \rho_0 { \ }d \mathcal{H}^2_x  \right) =0 .
\end{equation*}

An energetic variational approach is a method for deriving PDEs by using the forces derived from a variation of energies. Gyarmati \cite{Gya70} applied an energetic variational approach, which had been studied by Strutt \cite{Str73} and Onsager \cite{Ons31a, Ons31b}, to make several models for fluid dynamics in domains. Hyon-Kwak-Liu \cite{HKL10} made use of their energetic variational approach to study complex fluid in domains. Koba-Sato \cite{KS17} applied their energetic variational approach to make their non-Newtonian fluid systems in domains. Koba-Liu-Giga \cite{KLG17} and Koba \cite{K18} employed their energetic variational approaches to derive their fluid systems on an evolving closed surface. However, these papers \cite{HKL10,KS17,KLG17,K18} did not consider multiphase flow. This paper improves and modifies their methods in \cite{HKL10,KS17,KLG17,K18} to derive our multiphase flow systems. See Section \ref{sect4} for details.

Finally, we introduce the results related to this paper. Bothe-Pr\"{u}ss \cite{BP12,BP17} considered multiphase flow with interface effects. In \cite{BP12}, they made their models for multiphase flow with surface tension and viscosities by applying the \emph{Boussinesq-Scriven law}. In \cite{BP17}, they made use of their jump conditions to make models for multi-component two-phase flow system with phase transition, and to study the individual mass densities of an isothermal mixture of $N$-species in a domain. Although this paper does not consider surface viscosity (surface flow), this paper considers interface effects such as surface tension and phase transition. Note that our models are different from the ones in \cite{BP12, BP17}. See also \cite{KLG17, K18} for models for surface flow.

The outline of this paper is as follows: In Section \ref{sect2}, we state the main results of this paper. In Section \ref{sect3}, we study the law of conservation of mass for multiphase flow with phase transition. In Section \ref{sect4}, we apply an energetic variational approach to make mathematical models for multiphase flow with phase transition. In Section \ref{sect5}, we investigate the conservation and energy laws of our systems. In Appendix, we provide two useful lemmas to derive our systems.

\section{Main Results}\label{sect2}
We first introduce the transport theorems. Then we state the main results.
\begin{definition}[$\Omega_T$ is flowed by the velocity fields $(v_A,v_B,v_S)$]\label{def21} We say that $\Omega_T$ is \emph{flowed by the velocity fields} $(v_A,v_B ,v_S)$ if for each $0< t <T$, $f \in C^1 (\mathbb{R}^4)$, and $\Lambda \subset \Omega$,
\begin{align}
\label{eq21} \frac{d}{d t} \int_{\Omega_A (t) \cap \Lambda} f (x,t) { \ }d x & = \int_{\Omega_A (t) \cap \Lambda} \{ D_t^A f + ({\rm{div}}v_A )f \}{ \ }dx,\\ 
\frac{d}{d t} \int_{\Omega_B (t) \cap \Lambda} f (x,t) { \ }d x & = \int_{\Omega_B (t) \cap \Lambda} \{ D_t^B f + ({\rm{div}} v_B) f \} { \ }dx,\label{eq22}\\ 
\label{eq23} \frac{d}{d t} \int_{\Gamma (t) \cap \Lambda} f (x,t) { \ }d \mathcal{H}^2_x & = \int_{\Gamma (t) \cap \Lambda} \{ D_t^S f + ({\rm{div}}_\Gamma v_S) f \} { \ }d \mathcal{H}^2_x. 
\end{align}
Here $D_t^\sharp f = \partial_t f + (v_\sharp \cdot \nabla) f$, ${\rm{div}}_\Gamma v_S = \partial_1^\Gamma v^S_1 + \partial_2^\Gamma v^S_2 + \partial_3^\Gamma v^S_3$, $\partial_j^\Gamma f = \partial_j f - n^\Gamma_j (n_\Gamma \cdot \nabla )f $, where $\sharp = A,B,S$, and $j=1,2,3$. Note that ${\rm{div}} v_A =0$ in this paper.
\end{definition}
\noindent We often call \eqref{eq21}, \eqref{eq22} the \emph{transport theorems}, and \eqref{eq23} the \emph{surface transport theorem}. The derivation of the surface transport theorem can be founded in \cite{Bet86,GSW89,DE07,KLG17}. Throughout this paper we assume that $\Omega_T$ is \emph{flowed by the velocity fields} $(v_A,v_B ,v_S)$.

Now we state the main results of this paper.
\begin{theorem}[Laws of conservation of mass]\label{thm22}{ \ }\\
$(\rm{i})$ Assume that $(\rho_A, \rho_B, \rho_0, v_A, v_B, v_S, \pi_A, \pi_B,\pi_0)$ satisfy
\begin{equation}\label{eq24}
\begin{cases}
{\displaystyle{D_t^A \rho_A = {\rm{div}}\left\{ \frac{\rho_0}{\pi_0}\pi_A v_A \right\} }}& \text{ in } \Omega_{A,T},\\
{\displaystyle{D_t^B \rho_B + ({\rm{div}} v_B) \rho_B = {\rm{div}}\left\{ \frac{\rho_0}{\pi_0}\pi_B v_B \right\} }}& \text{ in } \Omega_{B,T},\\
\pi_0 H_\Gamma n_\Gamma - \pi_A n_\Gamma +   \pi_B n_\Gamma ={ }^t (0,0,0) & \text{ on } \Gamma_T,
\end{cases}
\end{equation}
and \eqref{eq12}. Then \eqref{eq17} holds for all $0 < t_1 < t_2 <T$.\\
\noindent $(\rm{ii})$ Assume that $(\rho_A, \rho_B, \rho_0, v_A, v_B, v_S, \pi_A, \pi_B,\pi_0, \mu_A, \mu_B, \lambda_B)$ satisfy
\begin{equation}\label{eq25}
\begin{cases}
{\displaystyle{D_t^A \rho_A  = - \frac{\rho_0}{\pi_0} {\rm{div}}( \mathcal{T}_A v_A) }}& \text{ in } \Omega_{A,T},\\
{\displaystyle{D_t^B \rho_B + ({\rm{div}} v_B) \rho_B = - \frac{\rho_0}{\pi_0} {\rm{div}} ( \mathcal{T}_B v_B )}} & \text{ in } \Omega_{B,T},\\
\pi_0 H_\Gamma n_\Gamma + \widetilde{\mathcal{T}}_A n_\Gamma - \widetilde{\mathcal{T}}_B n_\Gamma= { }^t( 0,0,0) & \text{ on } \Gamma_T,
\end{cases}
\end{equation}
and \eqref{eq15}, where $(\mathcal{T}_A, \mathcal{T}_B, \widetilde{\mathcal{T}}_A, \widetilde{\mathcal{T}}_B)$ are defined by \eqref{eq16}. Then \eqref{eq17} holds for all $0 < t_1 < t_2 < T$.
\end{theorem}

\begin{theorem}[Conservative form, conservation and energy Laws]\label{thm23}{ \ }\\
$(\rm{i})$ Any solution to system \eqref{eq13} with \eqref{eq12} satisfies \eqref{eq17} and \eqref{eq18} with $\mu_A \equiv \mu_B \equiv \lambda_B \equiv 0$.\\
$(\rm{ii})$ Any solution to system \eqref{eq14} with \eqref{eq15} satisfies \eqref{eq17} and \eqref{eq18}.\\
$(\rm{iii})$ If ${\rm{div}} v_A =0$ in $\Omega_{A,T}$, then system \eqref{eq14} satisfies the conservative form \eqref{eq19}.
\end{theorem}

We prove Theorem \ref{thm22} in Section \ref{sect3} and Theorem \ref{thm23} in Section \ref{sect5}. In Section \ref{sect4}, we derive our systems \eqref{eq13} and \eqref{eq14}.

\section{Laws of Conservation of Mass}\label{sect3}
Let us immediately derive the one of the main results of this paper.
\begin{proof}[Proof of Theorem \ref{thm22}]
We first show $(\rm{i})$. From \eqref{eq24}, we have
\begin{equation*}
\pi_0 H_\Gamma (v_S \cdot n_\Gamma) - \pi_A (v_S \cdot n_\Gamma) +  \pi_B (v_S \cdot n_\Gamma) = 0 \text{ on } \Gamma_T.
\end{equation*}
Since $\pi_0 \neq 0$ by assumption, we use \eqref{eq12} to derive
\begin{equation}\label{eq31}
H_\Gamma (v_S \cdot n_\Gamma) = \frac{1}{\pi_0} \pi_A (v_A \cdot n_\Gamma ) - \frac{1}{\pi_0} \pi_B (v_B \cdot n_\Gamma) \text{ on }\Gamma_T.  
\end{equation}
Using the transport theorems \eqref{eq21} and \eqref{eq22}, we check that
\begin{multline}\label{eq32}
\frac{d }{d t} \left( \int_{\Omega_A (t)} \rho_A (x,t) { \ }d x + \int_{\Omega_B(t)} \rho_B (x,t) { \ }d x + \int_{\Gamma (t)} \rho_0 { \ }d \mathcal{H}^2_x  \right)\\
 = \int_{\Omega_A (t)} D_t^A \rho_A { \ }d x + \int_{\Omega_B(t)} \{ D_t^B \rho_B + ({\rm{div}} v_B) \rho_B \} { \ }d x + \frac{d}{dt} \int_{\Gamma (t)}\rho_0{ \ }d \mathcal{H}^2_x.
\end{multline}
Applying the surface transport and divergence theorems \eqref{eq23}, \eqref{eq61} with \eqref{eq31}, we find that
\begin{multline*}
\frac{d}{dt}\int_{\Gamma (t)} \rho_0 { \ }d \mathcal{H}^2_x = \int_{\Gamma (t)} ({\rm{div}}_\Gamma v_S) \rho_0 { \ }d \mathcal{H}_x^2 = - \int_{\Gamma (t)} \rho_0 H_\Gamma (v_S \cdot n_\Gamma) { \ }d \mathcal{H}_x^2\\
= - \int_{\Gamma (t)} \frac{\rho_0}{\pi_0}\pi_A (v_A \cdot n_\Gamma ){ \ }d \mathcal{H}_x^2 - \int_{\partial \Omega} \frac{\rho_0}{\pi_0}\pi_B (v_B \cdot n_\Omega ){ \ }d \mathcal{H}_x^2 + \int_{\Gamma (t)} \frac{\rho_0}{\pi_0}\pi_B (v_B \cdot n_\Gamma ) { \ }d \mathcal{H}_x^2.
\end{multline*}
Note that $v_B \cdot n_\Omega =0$. Using the divergence theorem, we have
\begin{equation}\label{eq33}
\frac{d}{d t}\int_{\Gamma (t)} \rho_0 { \ }d \mathcal{H}_x^2 = - \int_{\Omega_A (t)}{\rm{div}}\left\{  \frac{\rho_0}{\pi_0} \pi_A v_A \right\} { \ }d x - \int_{\Omega_B (t)} {\rm{div}} \left\{ \frac{\rho_0}{\pi_0} \pi_B v_B \right\} { \ }d x.
\end{equation}
By \eqref{eq32}, \eqref{eq33}, and \eqref{eq24}, we see that
\begin{multline*}
\frac{d }{d t} \left( \int_{\Omega_A (t)} \rho_A(x,t) { \ }d x + \int_{\Omega_B(t) } \rho_B(x,t) { \ }d x + \int_{\Gamma (t)} \rho_0 { \ }d \mathcal{H}^2_x  \right)\\
 = \int_{\Omega_A (t)} \left( D_t^A \rho_A - {\rm{div}}\left\{  \frac{\rho_0}{\pi_0} \pi_A v_A \right\} \right)  { \ }d x\\ + \int_{\Omega_B(t)} \left(D_t^B \rho_B  + ({\rm{div}} v_B) \rho_B - {\rm{div}}\left\{  \frac{\rho_0}{\pi_0} \pi_B v_B \right\}  \right) { \ }d x = 0.
\end{multline*}
Integrating with respect to $t$, we have \eqref{eq17}. Therefore, we see $(\rm{i})$.

Before proving $(\rm{ii})$ we prepare the following lemma.
\begin{lemma}\label{lem31}If $P_\Gamma v_A = { }^t (0,0,0)$ and $P_\Gamma v_B = { }^t (0,0,0)$ on $\Gamma (t)$, then
\begin{align}
D (v_A) v_A \cdot n_\Gamma&= (n_\Gamma \cdot (n_\Gamma \cdot \nabla ) v_A) (v_A \cdot n_\Gamma ) \text{ on }\Gamma (t),\label{eq34}\\
D (v_B) v_B \cdot n_\Gamma &= (n_\Gamma \cdot (n_\Gamma \cdot \nabla ) v_B) (v_B \cdot n_\Gamma) \text{ on }\Gamma (t).\label{eq35}
\end{align}
\end{lemma}
\begin{proof}[Proof of Lemma \ref{lem31}]
We now drive \eqref{eq34}. Since $P_\Gamma v_A = { }^t (0,0,0)$, we find that
\begin{equation*}
v_A = P_\Gamma v_A + (v_A \cdot n_\Gamma )n_\Gamma = (v_A \cdot n_\Gamma )n_\Gamma \text{ on }\Gamma (t).
\end{equation*}
From $D (v_A) = \{ { }^t(\nabla v_A) + (\nabla v_A) \}/2$, we easily check that
\begin{align*}
D (v_A) v_A \cdot n_\Gamma & = \{ ( (v_A \cdot \nabla) v_A) \cdot n_\Gamma ) + ((n_\Gamma \cdot \nabla ) v_A) \cdot v_A \}/2\\
& = \{ ( ([(v_A \cdot n_\Gamma )n_\Gamma] \cdot \nabla) v_A) \cdot n_\Gamma ) + ((n_\Gamma \cdot \nabla )v_A) \cdot [(v_A \cdot n_\Gamma )n_\Gamma] \}/2\\
& = (n_\Gamma \cdot (n_\Gamma \cdot \nabla ) v_A) (v_A \cdot n_\Gamma) \text{ on }\Gamma (t),
\end{align*}
which is \eqref{eq34}. Similarly, we see \eqref{eq35}. Therefore, the lemma follows.
\end{proof}
Now we attack $(\rm{ii})$. By Lemma \ref{lem31} and \eqref{eq16}, we see that
\begin{equation}\label{eq36}
\begin{cases}
\mathcal{T}_A v_A \cdot n_\Gamma = \widetilde{\mathcal{T}}_A (v_A \cdot n_\Gamma) \text{ on } \Gamma (t),\\
\mathcal{T}_B v_B \cdot n_\Gamma = \widetilde{\mathcal{T}}_B (v_B \cdot n_\Gamma) \text{ on } \Gamma (t).
\end{cases}
\end{equation}
From \eqref{eq25} and \eqref{eq15}, we check that
\begin{align*}
\pi_0 H_\Gamma (v_S \cdot n_\Gamma) &= - \widetilde{\mathcal{T}}_A ( v_S \cdot n_\Gamma ) + \widetilde{\mathcal{T}}_B ( v_S \cdot n_\Gamma)\\
 &= - \widetilde{\mathcal{T}}_A (v_A \cdot n_\Gamma ) + \widetilde{\mathcal{T}}_B (v_B \cdot n_\Gamma) \text{ on } \Gamma (t).
\end{align*}
By \eqref{eq36}, we have
\begin{equation}\label{eq37}
H_\Gamma (v_S \cdot n_\Gamma )= - \frac{1}{\pi_0} (\mathcal{T}_A v_A \cdot n_\Gamma ) + \frac{1}{\pi_0} (\mathcal{T}_B v_B \cdot n_\Gamma) \text{ on }\Gamma (t).
\end{equation}
Applying the surface transport theorem \eqref{eq23}, the surface divergence theorem \eqref{eq61}, \eqref{eq37}, and $v_B \vert _{\partial \Omega} ={ }^t(0,0,0)$, we see that
\begin{multline*}
\frac{d}{d t} \int_{\Gamma (t)} \rho_0 { \ } d\mathcal{H}_x^2 = \int_{\Gamma (t)} \rho_0({\rm{div}}_\Gamma v_S) { \ } d \mathcal{H}_x^2= - \int_{\Gamma (t)} \rho_0 H_\Gamma (v_S \cdot n_\Gamma ) { \ } d \mathcal{H}_x^2\\
= \frac{\rho_0}{\pi_0}\int_{\Gamma (t)} \mathcal{T}_A v_A \cdot n_\Gamma { \ } d \mathcal{H}_x^2 + \frac{\rho_0}{\pi_0}\int_{\partial \Omega} \mathcal{T}_B v_B \cdot n_\Omega { \ } d \mathcal{H}_x^2 - \frac{\rho_0}{\pi_0}\int_{\Gamma (t)} \mathcal{T}_B v_B \cdot n_\Gamma { \ } d \mathcal{H}_x^2.
\end{multline*}
Using the divergence theorem, we check that
\begin{equation*}
\frac{d}{d t} \int_{\Gamma (t)} \rho_0 { \ } d\mathcal{H}_x^2 = \int_{\Omega_A (t)}\frac{\rho_0}{\pi_0} {\rm{div}}( \mathcal{T}_A v_A ) { \ } d x + \int_{\Omega_B (t)}\frac{\rho_0}{\pi_0} {\rm{div}} (\mathcal{T}_B v_B) { \ } d x.
\end{equation*}
By the same argument as in $(\rm{i})$, we see $(\rm{ii})$. Therefore, Theorem \ref{thm22} is proved.
\end{proof}

\section{Mathematical Modeling}\label{sect4}

In this section we make mathematical models for multiphase flow with phase transition. We apply our energetic variational approaches to derive system \eqref{eq14} in subsection \ref{subsec41} and system \eqref{eq13} in subsection \ref{subsec42}. 
\subsection{Viscous Model}\label{subsec41}
Under the restriction \eqref{eq15}, we apply an energetic variational approach to derive system \eqref{eq14}. We assume that $(v_A , v_B, v_S)$ satisfies \eqref{eq15}. 

Let $\Phi_A, \Phi_B \in C (\mathbb{R}^4)$. We assume that the dominant equations for the densities of our system are written by
\begin{equation}\label{eq41}
\begin{cases}
D_t^A \rho_A = \Phi_A & \text{ in } \Omega_{A,T},\\
D_t^B \rho_B + ({\rm{div}} v_B) \rho_B = \Phi_B & \text{ in } \Omega_{B,T}.
\end{cases}
\end{equation}
From now we look for $\Phi_A$, $\Phi_B$ satisfying
\begin{equation*}
\frac{d }{d t} \left( \int_{\Omega_A (t)} \rho_A (x,t) { \ }d x + \int_{\Omega_B(t)} \rho_B (x,t) { \ }d x + \int_{\Gamma (t)} \rho_0 { \ }d \mathcal{H}^2_x  \right) =0
\end{equation*}
by applying an energetic variational approach.

In order to derive the momentum equations of our system, we now discuss the variation of the velocities $(v_A,v_B,v_S)$ to the work and dissipation energies for our viscous model. Fix $0<t<T$. We set the work $E_W$ done by pressures $\pi_B$ and $\pi_0$, and the dissipation energies $E_D$ due to viscosities $(\mu_A, \mu_B , \lambda_B)$ as follows:
\begin{equation*}
E_W [v_A, v_B , v_S] =  \int_{\Omega_A (t)} ({\rm{div}} v_B) \pi_B { \ } d x + \int_{\Gamma(t)} ({\rm{div}}_\Gamma v_S) \pi_0 { \ }d \mathcal{H}^2_x,
\end{equation*}
\begin{multline*}
E_{D} [v_A, v_B , v_S] = \int_{\Omega_A (t)} \left( - \frac{\mu_A}{2} \vert D (v_A) \vert^2\right){ \ } d x\\ + \int_{\Omega_B (t)} \left( - \frac{\mu_B}{2} \vert D (v_B) \vert^2 - \frac{\lambda_{B}}{2} \vert {\rm{div}} v_B \vert^2 \right){ \ } d x.
\end{multline*}
Set $E_{D+W}[\cdot] = E_D[\cdot] + E_W[\cdot]$. 
\begin{remark}\label{rem41}$(\rm{i})$ From ${\rm{div}} v_A =0$ we see that $({\rm{div}} v_A) \pi_A =0$.\\
$(\rm{ii})$ Collectively, we call $ ({\rm{div}}v_B) \pi_B $, $({\rm{div}}_\Gamma v_S) \pi_0$, $\mu_A  \vert D (v_A) \vert ^2$, $\mu_B \vert D (v_B) \vert ^2$,  $\lambda_B  \vert {\rm{div}}v_B \vert ^2$ the \emph{energy densities}. See \cite{KS17} and \cite{K18} for mathematical validity of the energy densities.
\end{remark}
We consider the variation of $E_{D+W}$ with respect to the velocities $(v_A , v_B ,v_S)$. Let $\varphi_A, \varphi_B, \varphi_S \in [C^\infty (\mathbb{R}^3)]^3$. For $- 1 < \varepsilon < 1$, $v_A^\varepsilon := v_A + \varepsilon \varphi_A$, $v_B^\varepsilon := v_B + \varepsilon \varphi_B$, $v_S^\varepsilon := v_S + \varepsilon \varphi_S$. We call ($v_A^\varepsilon$, $v_B^\varepsilon$, $v_S^\varepsilon$) a variation of ($v_A$, $v_B$, $v_S$). For each variation $(v_A^\varepsilon, v_B^\varepsilon, v_S^\varepsilon)$,
\begin{multline*}
E_{D+W} [v_A^\varepsilon, v_B^\varepsilon, v_S^\varepsilon] := \int_{\Omega_B (t)} \left\{ ({\rm{div}} v^\varepsilon_B) \pi_B - \frac{\mu_B}{2} \vert D (v^\varepsilon_B) \vert^2- \frac{\lambda_{B}}{2} \vert {\rm{div}} v^\varepsilon_B  \vert^2 \right\}{ \ } d x\\
+ \int_{\Omega_A (t)} \left\{ - \frac{\mu_A}{2} \vert D (v^\varepsilon_A) \vert^2\right\}{ \ } d x + \int_{\Gamma(t)} ({\rm{div}}_\Gamma v^\varepsilon_S) \pi_0 { \ }d \mathcal{H}^2_x.
\end{multline*}
A direct calculation gives
\begin{multline*}
\frac{d}{d \varepsilon} \bigg \vert _{\varepsilon =0 }E_{D+W} [v^\varepsilon_A, v^\varepsilon_B , v^\varepsilon_S]\\
 =  \int_{\Omega_B (t)}\{ ({\rm{div}} \varphi_B) \pi_B - \mu_B D(v_B): D (\varphi_B) - \lambda_B ({\rm{div}} v_B)({\rm{div}} \varphi_B) \} { \ } d x\\
 + \int_{\Omega_A(t)} \{ - \mu_A D (v_A): D (\varphi_A) \} { \ }d x + \int_{\Gamma(t)} ({\rm{div}}_\Gamma \varphi_S) \pi_0 { \ }d \mathcal{H}^2_x.
\end{multline*}
From \eqref{eq15}, we assume that for $- 1  < \varepsilon < 1$,
\begin{equation*}
\begin{cases}
{\rm{div}} v_A^\varepsilon =0 & \text{ in } \Omega_A (t),\\
v^\varepsilon_A \cdot n_\Gamma = v^\varepsilon_S \cdot n_\Gamma & \text{ on } \Gamma (t),\\
v^\varepsilon_B \cdot n_\Gamma = v^\varepsilon_S \cdot n_\Gamma  & \text{ on } \Gamma (t),
\end{cases}{ \ }\begin{cases}
v^\varepsilon_B = { }^t (0,0,0) & \text{ on }\partial \Omega,\\
P_\Gamma v^\varepsilon_A = { }^t (0,0,0) & \text{ on } \Gamma (t),\\
P_\Gamma v^\varepsilon_B = { }^t (0,0,0) & \text{ on } \Gamma (t).
\end{cases}
\end{equation*}
Then we have
\begin{equation}\label{eq42}
\begin{cases}
{\rm{div}} \varphi_A =0 & \text{ in } \Omega_A (t),\\
\varphi_A \cdot n_\Gamma = \varphi_S \cdot n_\Gamma & \text{ on } \Gamma (t),\\
\varphi_B \cdot n_\Gamma = \varphi_S \cdot n_\Gamma  & \text{ on } \Gamma (t),
\end{cases}{ \ }\begin{cases}
\varphi_B = { }^t (0,0,0) & \text{ on }\partial \Omega,\\
P_\Gamma \varphi_A = { }^t (0,0,0) & \text{ on } \Gamma (t),\\
P_\Gamma \varphi_B = { }^t (0,0,0) & \text{ on } \Gamma (t).
\end{cases}
\end{equation}

Now we study the forces derived from a variation of the work and dissipation energies.
\begin{lemma}\label{lem42}
Let $0<t<T$, and $F_A , F_B, F_S \in [C (\mathbb{R}^3) ]^3$. Assume that for every $\varphi_A , \varphi_B, \varphi_S \in [C^\infty (\mathbb{R}^3) ]^3$ satisfying \eqref{eq42},
\begin{equation*}
\frac{d}{d \varepsilon } \bigg\vert_{\varepsilon =0} E_{D + W} [v_A^\varepsilon, v_B^\varepsilon , v_S^\varepsilon] = \int_{\Omega_A (t)} F_A \cdot \varphi_A { \ } d x + \int_{\Omega_B (t)} F_B \cdot \varphi_B { \ } d x + \int_{\Gamma (t)} F_S \cdot \varphi_S { \ } d \mathcal{H}^2_x.
\end{equation*}
Then there is $\pi_A \in C^1 (\overline{\Omega_A (t)})$ such that
\begin{equation}\label{eq43}
\begin{cases}
F_A = {\rm{div}} \mathcal{T}_A (v_A ,\pi_A ) & \text{ in } \Omega_{A} (t),\\
F_B = {\rm{div}} \mathcal{T}_B (v_B , \pi_B ) & \text{ in } \Omega_{B} (t),\\
F_S = - \pi_0 H_\Gamma n_\Gamma - \widetilde{\mathcal{T} }_A(v_A, \pi_A) n_\Gamma + \widetilde{\mathcal{T}}_B(v_B , \pi_B) n_\Gamma & \text{ on } \Gamma (t).
\end{cases}
\end{equation}
Here $(\mathcal{T}_A, \mathcal{T}_B, \widetilde{\mathcal{T}}_A, \widetilde{\mathcal{T}}_B)$ is defined by \eqref{eq16}.
\end{lemma}

\begin{proof}[Proof of Lemma \ref{lem42}]
By assumption, we find that for all $\varphi_A , \varphi_B, \varphi_S \in [C^\infty (\mathbb{R}^3) ]^3$ satisfying \eqref{eq42},
\begin{multline*}
 \int_{\Omega_B (t)}\{ ({\rm{div}} \varphi_B) \pi_B - \mu_B D(v_B): D (\varphi_B) - \lambda_B ({\rm{div}} v_B)({\rm{div}} \varphi_B) \} { \ } d x\\
 + \int_{\Omega_A(t)} \{ - \mu_A D (v_A): D (\varphi_A) \} { \ }d x + \int_{\Gamma(t)} ({\rm{div}}_\Gamma \varphi_S) \pi_0 { \ }d \mathcal{H}^2_x\\
 = \int_{\Omega_A (t)} F_A \cdot \varphi_A { \ } d x + \int_{\Omega_B (t)} F_B \cdot \varphi_B { \ } d x + \int_{\Gamma (t)} F_S \cdot \varphi_S { \ } d \mathcal{H}^2_x.
\end{multline*}
Using the surface divergence theorem \eqref{eq61}, integration by parts, and \eqref{eq42}, we have
\begin{multline}\label{eq44}
\int_{\Omega_B (t)} ( - F_B - \nabla \pi_B + {\rm{div}}\{ \mu_B D (v_B ) + \lambda_B ({\rm{div}} v_B ) I_{3 \times 3} \}) \cdot \varphi_B { \ } d x\\
 + \int_{\Omega_A(t)} (- F_A + {\rm{div}}\{ \mu_A D (v_A) \}) \cdot \varphi_A { \ }d x\\
 + \int_{\Gamma(t)} \{ - F_S - \pi_0 H_\Gamma n_\Gamma + \widetilde{\mathcal{T}}_B n_\Gamma - \mu_A (n_\Gamma \cdot ( n_\Gamma \cdot \nabla ) v_A ) n_\Gamma \} \cdot \varphi_S { \ }d \mathcal{H}^2_x = 0.
\end{multline}
Here we used the facts that
\begin{align*}
D (v_A) n_\Gamma \cdot \varphi_A &= (n_\Gamma \cdot (n_\Gamma \cdot \nabla ) v_A) (n_\Gamma \cdot \varphi_A) \text{ on }\Gamma (t),\\
D (v_B) n_\Gamma \cdot \varphi_B &= (n_\Gamma \cdot (n_\Gamma \cdot \nabla ) v_B) (n_\Gamma \cdot \varphi_B) \text{ on }\Gamma (t).
\end{align*}

We now consider the case when  $\varphi_A = { }^t (0,0,0)$ and $\varphi_S = { }^t (0,0,0)$, that is, for every $\varphi_B \in [C^\infty (\mathbb{R}^3)]^3$ satisfying $\varphi_B = { }^t (0,0,0)$ on $\partial \Omega$ and $\varphi_B = { }^t (0,0,0)$ on $\Gamma (t)$,
\begin{equation*}
\int_{\Omega_B (t)} ( - F_B - \nabla \pi_B + {\rm{div}}\{ \mu_B D (v_B ) + \lambda_B ({\rm{div}} v_B ) I_{3 \times 3} \}) \cdot \varphi_B { \ } d x = 0.
\end{equation*}
This shows that
\begin{equation}\label{eq45}
F_B = {\rm{div}} \mathcal{T}_B ( v_B , \pi_B )  \text{ in }\Omega_B (t).
\end{equation}

Next we consider the case when  $\varphi_B = { }^t (0,0,0)$ and $\varphi_S = { }^t (0,0,0)$, that is, for every $\varphi_A \in [C^\infty (\mathbb{R}^3)]^3$ satisfying $\varphi_A = { }^t (0,0,0)$ on $\Gamma (t)$ and ${\rm{div}} \varphi _A =0$ in $\Omega_A (t)$,
\begin{equation*}
\int_{\Omega_A (t)} ( - F_A + {\rm{div}}\{\mu_A D (v_A) \}  ) \cdot \varphi_A { \ } d x =0.
\end{equation*}
Since ${\rm{div}} \varphi_A =0$ in $\Omega_A (t)$, we apply the Helmholtz-Weyl decomposition (Lemma \ref{lem62}) to find that there exists $\pi_A \in C^1(\overline{\Omega_A (t)})$ such that
\begin{equation}\label{eq46}
- F_A + {\rm{div}}\{\mu_A D (v_A) \}  = \nabla \pi_A \text{ in } \Omega_A (t).
\end{equation}

Finally, we consider the case when $\varphi_S \neq { }^t (0,0,0)$. By \eqref{eq44}, \eqref{eq45}, \eqref{eq46}, we see that
\begin{multline*}
\int_{\Omega_A (t)} (\nabla \pi_A) \cdot \varphi_A { \ } d x\\ + \int_{\Gamma(t)} \{ - F_S - \pi_0 H_\Gamma n_\Gamma + \widetilde{\mathcal{T}}_B n_\Gamma - \mu_A (n_\Gamma \cdot ( n_\Gamma \cdot \nabla ) v_A ) n_\Gamma \} \cdot \varphi_S { \ }d \mathcal{H}^2_x = 0.
\end{multline*}
Using integration by parts with ${\rm{div}} \varphi_A = 0$ and $\varphi_A \cdot n_\Gamma = \varphi_S \cdot n_\Gamma$, we have
\begin{equation*}
\int_{\Gamma(t)} ( - F_S - \pi_0 H_\Gamma n_\Gamma - \widetilde{\mathcal{T}}_A n_\Gamma + \widetilde{\mathcal{T}}_B n_\Gamma ) \cdot \varphi_S { \ }d \mathcal{H}^2_x = 0.
\end{equation*}
Since the above equality holds for all $\varphi_S \in [C^\infty (\mathbb{R}^3)]^3$, we see that
\begin{equation}\label{eq47}
F_S = - \pi_0 H_\Gamma n_\Gamma - \widetilde{\mathcal{T}}_A n_\Gamma +  \widetilde{\mathcal{T}}_B n_\Gamma \text{ on } \Gamma (t).
\end{equation}
Therefore, Lemma \ref{lem42} is proved.
\end{proof}

Now we return to derive our momentum equations. We admit the \emph{fundamental principle of the dynamics of fluid motion} (see Chapter B in \cite{Ser59}). We assume that the time rate of change of the momentum equals to the forces derived from the variation of the work done by pressures and the energies dissipation due to viscosities, that is, suppose that for every $0 < t <T$ and $\Lambda \subset \Omega$,
\begin{align*}
\frac{d}{d t} \int_{\Omega_A (t) \cap \Lambda} \rho_A v_A { \ }d x = \int_{\Omega_A (t) \cap \Lambda} F_A { \ } d x,\\
\frac{d}{d t} \int_{\Omega_B (t) \cap \Lambda} \rho_B v_B  { \ }d x = \int_{\Omega_B (t) \cap \Lambda} F_B { \ } d x,\\
0 = \int_{\Gamma (t) \cap \Lambda} F_S { \ } d \mathcal{H}^2_x.
\end{align*}
Here $(F_A, F_B ,F_S)$ is defined by \eqref{eq43}. Remark that we do not consider the momentum on the surface $\Gamma (t)$ since we do not consider surface flow. Applying the transport theorems with \eqref{eq41}, we have
\begin{equation}\label{eq48}
\begin{cases}
\rho_A D_t^A v_A + \Phi_A v_A = F_A & \text{ in }\Omega_{A,T},\\
\rho_B D_t^B v_B + \Phi_B v_B = F_B & \text{ in }\Omega_{B,T},\\
\pi_0 H_\Gamma n_\Gamma + \widetilde{\mathcal{T}}_A n_\Gamma - \widetilde{\mathcal{T}}_B n_\Gamma = { }^t (0,0,0) & \text{ in }\Gamma_T.
\end{cases}
\end{equation}
Using the transport theorems \eqref{eq21}-\eqref{eq23} with \eqref{eq41}, we check that
\begin{multline*}
\frac{d }{d t} \left( \int_{\Omega_A (t)} \rho_A (x,t) { \ }d x + \int_{\Omega_B(t)} \rho_B (x,t) { \ }d x + \int_{\Gamma (t)} \rho_0 { \ }d \mathcal{H}^2_x  \right)\\
 = \int_{\Omega_A (t)} \Phi_A  { \ }d x + \int_{\Omega_B(t)} \Phi_B { \ }d x + \int_{\Gamma (t)} \rho_0 ({\rm{div}}_\Gamma v_S) { \ }d \mathcal{H}^2_x.
\end{multline*}
By the same argument in the proof of Theorem \ref{thm22}, we see that
\begin{multline*}
\frac{d }{d t} \left( \int_{\Omega_A (t)} \rho_A (x,t) { \ }d x + \int_{\Omega_B(t)} \rho_B (x,t) { \ }d x + \int_{\Gamma (t)} \rho_0 { \ }d \mathcal{H}^2_x  \right)\\
 = \int_{\Omega_A (t)} \left( \Phi_A + \frac{\rho_0}{\pi_0} {\rm{div}}(\mathcal{T}_A v_A) \right){ \ }d x + \int_{\Omega_B(t)} \left( \Phi_B + \frac{\rho_0}{\pi_0} {\rm{div}}(\mathcal{T}_B v_B) \right){ \ }d x.
\end{multline*}
Thus, we set  $\Phi_A = - \frac{\rho_0}{\pi_0} {\rm{div}}(\mathcal{T}_A v_A)$ and $ \Phi_B = - \frac{\rho_0}{\pi_0}{\rm{div}}(\mathcal{T}_B v_B)$ to derive
\begin{equation*}
\frac{d }{d t} \left( \int_{\Omega_A (t)} \rho_A (x,t) { \ }d x + \int_{\Omega_B(t)} \rho_B (x,t) { \ }d x + \int_{\Gamma (t)} \rho_0 { \ }d \mathcal{H}^2_x  \right)=0.
\end{equation*}
Therefore, combining \eqref{eq41}, \eqref{eq43}, and \eqref{eq48}, we have system \eqref{eq14}.

\subsection{Inviscid Model}\label{subsec42}
Under the restriction \eqref{eq12} we apply an energetic variational approach to derive system \eqref{eq13}. We assume that $(v_A,v_B,v_S)$ satisfy \eqref{eq12}.

Let $\Psi_A, \Psi_B \in C (\mathbb{R}^4)$. We assume that the dominant equations for the densities of our system are written by
\begin{equation}\label{eq49}
\begin{cases}
D_t^A \rho_A  = \Psi_A & \text{ in } \Omega_{A,T},\\
D_t^B \rho_B + ({\rm{div}} v_B) \rho_B = \Psi_B & \text{ in } \Omega_{B,T}.
\end{cases}
\end{equation}
From now we look for $\Psi_A$, $\Psi_B$ satisfying
\begin{equation*}
\frac{d }{d t} \left( \int_{\Omega_A (t)} \rho_A (x,t) { \ }d x + \int_{\Omega_B(t)} \rho_B (x,t) { \ }d x + \int_{\Gamma (t)} \rho_0 { \ }d \mathcal{H}^2_x  \right) =0
\end{equation*}
by applying an energetic variational approach.

In order to derive the momentum equations of our system, we now discuss the variation of the velocities $(v_A,v_B,v_S)$ to the work for our inviscid model. Fix $0<t<T$. We set the work $E_W$ done by pressures $\pi_B$ and $\pi_0$ as follows:
\begin{equation*}
E_W [v_A, v_B , v_S] =  \int_{\Omega_B (t)} ({\rm{div}} v_B) \pi_B { \ } d x + \int_{\Gamma(t)} ({\rm{div}}_\Gamma v_S) \pi_0 { \ }d \mathcal{H}^2_x.
\end{equation*}
We consider the variation of the work $E_W$ with respect to the velocities $(v_A , v_B ,v_S)$. Let $0<t<T$. Let $\varphi_A, \varphi_B, \varphi_S \in [C^\infty (\mathbb{R}^3)]^3$. For $- 1 < \varepsilon < 1$, $v_A^\varepsilon := v_A + \varepsilon \varphi_A$, $v_B^\varepsilon := v_B + \varepsilon \varphi_B$, $v_S^\varepsilon := v_S + \varepsilon \varphi_S$. We call ($v_A^\varepsilon$, $v_B^\varepsilon$, $v_S^\varepsilon$) a variation of ($v_A$, $v_B$, $v_S$). For each variation $(v_A^\varepsilon, v_B^\varepsilon, v_S^\varepsilon)$,
\begin{equation*}
E_W [v_A^\varepsilon, v_B^\varepsilon, v_S^\varepsilon] :=  \int_{\Omega_B (t)} ({\rm{div}} v^\varepsilon_B) \pi_B { \ } d x + \int_{\Gamma(t)} ({\rm{div}}_\Gamma v^\varepsilon_S) \pi_0 { \ }d \mathcal{H}^2_x.
\end{equation*}
From \eqref{eq12}, we assume that for $- 1  < \varepsilon < 1$,
\begin{equation*}
\begin{cases}
{\rm{div}} v_A^\varepsilon =0 & \text{ in } \Omega_A (t),\\
v^\varepsilon_B \cdot n_\Omega =0 & \text{ on }\partial \Omega,
\end{cases}{ \ }\begin{cases}
v^\varepsilon_A \cdot n_\Gamma =  v^\varepsilon_S \cdot n_\Gamma & \text{ on } \Gamma (t),\\
v^\varepsilon_B \cdot n_\Gamma = v^\varepsilon_S \cdot n_\Gamma  & \text{ on } \Gamma (t).
\end{cases}
\end{equation*}
Then we have
\begin{equation}\label{eq4010}
\begin{cases}
{\rm{div}} \varphi_A =0 & \text{ in } \Omega_A (t),\\
\varphi_B \cdot n_\Omega =0 & \text{ on }\partial \Omega,
\end{cases}{ \ }\begin{cases}
\varphi_A \cdot n_\Gamma =  \varphi_S \cdot n_\Gamma & \text{ on } \Gamma (t),\\
\varphi_B \cdot n_\Gamma =  \varphi_S \cdot n_\Gamma  & \text{ on } \Gamma (t).
\end{cases}
\end{equation}

Now we study the forces derived from a variation of the work.
\begin{lemma}\label{lem44}
Let $0<t<T$, and $G_A , G_B, G_S \in [C (\mathbb{R}^3) ]^3$. Assume that for every $\varphi_A , \varphi_B, \varphi_S \in [C^\infty (\mathbb{R}^3) ]^3$ satisfying \eqref{eq4010},
\begin{equation*}
\frac{d}{d \varepsilon } \bigg\vert_{\varepsilon =0} E_{W} [v_A^\varepsilon, v_B^\varepsilon , v_S^\varepsilon] = \int_{\Omega_A (t)} G_A \cdot \varphi_A { \ } d x + \int_{\Omega_B (t)} G_B \cdot \varphi_B { \ } d x + \int_{\Gamma (t)} G_S \cdot \varphi_S { \ } d \mathcal{H}^2_x.
\end{equation*}
Then there is $\pi_A \in C^1 (\overline{\Omega_A (t)})$ such that
\begin{equation}\label{eq4011}
\begin{cases}
G_A = - {\rm{grad}}\pi_A & \text{ in } \Omega_{A} (t),\\
G_B = - {\rm{grad}} \pi_B & \text{ in } \Omega_{B} (t),\\
G_S = - \pi_0 H_\Gamma n_\Gamma + \pi_A n_\Gamma - \pi_B n_\Gamma & \text{ on } \Gamma (t).
\end{cases}
\end{equation}
\end{lemma}
\noindent By the same arguments in the proof of Lemma \ref{lem42}, we can prove Lemma \ref{lem44}.

Now we derive our momentum equations. We admit the \emph{principle of conservation of linear momentum} (see Chapter B in \cite{Ser59}). We assume that the time rate of change of the momentum equals to the forces derived from the work done by pressures, that is, suppose that for every $0 < t <T$ and $\Lambda \subset \Omega$,
\begin{align*}
\frac{d}{d t} \int_{\Omega_A (t) \cap \Lambda} \rho_A v_A { \ }d x = \int_{\Omega_A (t) \cap \Lambda} G_A { \ } d x,\\
\frac{d}{d t} \int_{\Omega_B (t) \cap \Lambda} \rho_B v_B  { \ }d x = \int_{\Omega_B (t) \cap \Lambda} G_B { \ } d x,\\
0 = \int_{\Gamma (t) \cap \Lambda} G_S { \ } d \mathcal{H}^2_x.
\end{align*}
Here $( G_A , G_B , G_S )$ is defined by \eqref{eq4011}. Remark that we do not consider the momentum on the surface $\Gamma (t)$ since we do not consider surface flow. Applying the transport theorems with \eqref{eq49}, we have
\begin{equation}\label{eq4012}
\begin{cases}
\rho_A D_t^A v_A + \Psi_A v_A = G_A & \text{ in }\Omega_{A,T},\\
\rho_B D_t^B v_B + \Psi_B v_B = G_B & \text{ in }\Omega_{B,T} ,\\
\pi_0 H_\Gamma n_\Gamma -  \pi_A n_\Gamma +  \pi_B n_\Gamma = { }^t (0,0,0) & \text{ in }\Gamma_T.
\end{cases}
\end{equation}
Using the transport theorems \eqref{eq21}-\eqref{eq23} with \eqref{eq49}, we check that
\begin{multline*}
\frac{d }{d t} \left( \int_{\Omega_A (t)} \rho_A (x,t) { \ }d x + \int_{\Omega_B(t)} \rho_B (x,t) { \ }d x + \int_{\Gamma (t)} \rho_0 { \ }d \mathcal{H}^2_x  \right)\\
 = \int_{\Omega_A (t)} \Psi_A  { \ }d x + \int_{\Omega_B(t)} \Psi_B { \ }d x + \int_{\Gamma (t)} \rho_0({\rm{div}}_\Gamma v_S) { \ }d \mathcal{H}^2_x.
\end{multline*}
By the same argument in the proof of Theorem \ref{thm22}, we see that
\begin{multline*}
\frac{d }{d t} \left( \int_{\Omega_A (t)} \rho_A (x,t) { \ }d x + \int_{\Omega_B(t)} \rho_B (x,t) { \ }d x + \int_{\Gamma (t)} \rho_0 { \ }d \mathcal{H}^2_x  \right)\\
 = \int_{\Omega_A (t)} \left(\Psi_A - {\rm{div}} \left\{ \frac{\rho_0}{\pi_0} \pi_A v_A \right\}  \right){ \ }d x + \int_{\Omega_B(t)} \left( \Psi_B - {\rm{div}} \left\{ \frac{\rho_0}{\pi_0} \pi_B v_B \right\} \right){ \ }d x.
\end{multline*}
Thus, we set  $\Psi_A = {\rm{div}} \{ \frac{\rho_0}{\pi_0} \pi_A v_A \}$ and $ \Psi_B = {\rm{div}} \{ \frac{\rho_0}{\pi_0} \pi_B v_B \}$ to derive
\begin{equation*}
\frac{d }{d t} \left( \int_{\Omega_A (t)} \rho_A (x,t) { \ }d x + \int_{\Omega_B(t)} \rho_B (x,t) { \ }d x + \int_{\Gamma (t)} \rho_0 { \ }d \mathcal{H}^2_x  \right)=0.
\end{equation*}
Combining \eqref{eq49}, \eqref{eq4011}, and \eqref{eq4012}, therefore, we have system \eqref{eq13}.

\section{Conservation and Energy Laws}\label{sect5}

Applying the transport theorems and integration by parts, we can prove Theorem \ref{thm23}. Therefore, we only show that any solution to system \eqref{eq14} with \eqref{eq15} satisfies \eqref{eq18}.

Applying the transport theorems \eqref{eq21}, \eqref{eq22}, and system \eqref{eq14}, we see that
\begin{multline}\label{eq51}
\frac{d}{d t} \left( \int_{\Omega_A (t)} \frac{1}{2} \rho_A  \vert v_A \vert^2{ \ }d x + \int_{\Omega_B (t)} \frac{1}{2}\rho_B  \vert v_B \vert^2{ \ } d x \right)\\
=  \int_{\Omega_A (t)} \left( \frac{1}{2}  \vert v_A \vert^2( D_t^A \rho_A) + \rho_A D_t^A v_A \cdot v_A \right) { \ }d x\\
+ \int_{\Omega_B (t)} \left( \frac{1}{2}  \vert v_B \vert^2( D_t^B \rho_B + ({\rm{div}} v_B) \rho_B ) + \rho_B D_t^B v_B \cdot v_B \right) { \ }d x\\
= \int_{\Omega_A (t)}\left(  ({\rm{div}}\mathcal{T}_A) \cdot v_A + \frac{\rho_0}{2\pi_0}{\rm{div}} (\mathcal{T}_A v_A)   \vert v_A \vert^2\right) { \ }d x\\
+ \int_{\Omega_B (t)} \left( ({\rm{div}}\mathcal{T}_B) \cdot v_B + \frac{\rho_0}{2 \pi_0} {\rm{div}}(\mathcal{T}_B v_B)  \vert v_B \vert^2\right) { \ }d x.
\end{multline}
Using integration by parts with $v_B  \vert _{\partial \Omega} = { }^t (0,0,0)$, we observe that
\begin{multline*}
\text{(R.H.S.) of } \eqref{eq51}\\ = \int_{\Omega_B (t)}\left( ({\rm{div}} v_B) \pi_B + \frac{\rho_0}{2\pi_0} {\rm{div}}(\mathcal{T}_B v_B )  \vert v_B \vert^2 - \mu_B  \vert D (v_B) \vert^2- \lambda_B  \vert  {\rm{div}} v_B  \vert^2\right) { \ }d x\\
+ \int_{\Omega_A (t)} \left( \frac{\rho_0}{2 \pi_0} {\rm{div}}(\mathcal{T}_A v_A)  \vert v_A \vert^2- \mu_A \vert  D (v_A)  \vert^2\right) { \ }d x + \int_{\Gamma (t)} ({\rm{div}}_\Gamma v_S ) \pi_0{ \ }d \mathcal{H}_x^2.
\end{multline*}
Here we used the fact that
\begin{align*}
\int_{\Gamma (t)} (\mathcal{T}_A n_\Gamma \cdot v_A - \mathcal{T}_B n_\Gamma \cdot v_B) { \ }d \mathcal{H}^2_x = \int_{\Gamma (t)} ( \widetilde{\mathcal{T}}_A -  \widetilde{\mathcal{T}}_B)  v_S \cdot n_\Gamma { \ }d \mathcal{H}^2_x\\
 = - \int_{\Gamma (t)} \pi_0 H_\Gamma (v_S \cdot n_\Gamma ) { \ }d \mathcal{H}^2_x = \int_{\Gamma (t)} ({\rm{div}}_\Gamma v_S ) \pi_0 { \ }d \mathcal{H}_x^2.
\end{align*}
Integrating with respect to $t$, we have \eqref{eq18}. Therefore, Theorem \ref{thm23} is proved.

\section{Appendix: Tools}\label{sect6}
We introduce useful tools to make a mathematical model for multiphase flow.
\begin{lemma}[Surface divergence theorem]\label{lem61}
Let $\Gamma_* \subset \mathbb{R}^3$ be a smooth closed $2$-dimensional surface. Then for each $V_S \in [C^1 (\Gamma_*)]^3$,
\begin{equation}\label{eq61}
\int_{\Gamma_*} {\rm{div}}_{\Gamma_*} V_S { \ } d \mathcal{H}^2_x = - \int_{\Gamma_*} H_{\Gamma_*} ( V_S \cdot n_{\Gamma_*}) { \ } d \mathcal{H}^2_x.
\end{equation}
Here $H_{\Gamma_*}$ is the mean curvature in the direction $n_{\Gamma_*}$ defined by $H_{\Gamma_*} = - {\rm{div}}_{\Gamma_*} n_{\Gamma_*}$, where $n_{\Gamma_*} = n_{\Gamma_*} (x) = { }^t (n^*_1 , n^*_2 , n^*_3)$ denotes the unit outer normal vector at $x \in \Gamma_*$.
\end{lemma}
\noindent The proof of Lemma \ref{lem61} can be founded in Simon \cite{Sim83} and Koba \cite{K19}.

\begin{lemma}[Helmholtz-Weyl decomposition]\label{lem62}
Let $\Omega_*$ be a bounded domain in $\mathbb{R}^3$ with smooth boundary $\partial \Omega_*$. Set
\begin{equation*}
C^\infty_{0,{\rm{div}}}( \Omega_*) = \{ \varphi \in [ C_0^\infty (\Omega_*) ]^3;  { \ }{\rm{div}} \varphi = 0 \}.
\end{equation*}
Let $F_* \in [C(\overline{\Omega_*})]^3$. Assume that for each $\varphi \in C^\infty_{0,{\rm{div}}}( \Omega_*)$
\begin{equation*}
\int_{\Omega_*} F_* \cdot \varphi { \ }d x = 0.
\end{equation*}
Then there is $\Pi_{*} \in C^1 (\overline{\Omega_*})$ such that $F_* = \nabla \Pi_*$ in $\Omega_*$.\\
\end{lemma}
\noindent The proof of Lemma \ref{lem62} can be founded in Temam \cite{Tem77} and Sohr \cite{Soh01}.

\section*{Acknowledgments}

This work was partly supported by the Japan Society for the Promotion of Science (JSPS) KAKENHI Grant Number JP21K03326.

\bibliographystyle{siamplain}
\bibliography{references}

\end{document}